%% file: arxiv.tex
\documentclass[a4paper,USenglish,cleveref,autoref,thm-restate]{lipics-v2021}

\pdfoutput=1 % uncomment to ensure pdflatex processing (mandatatory
\hideLIPIcs % uncomment to remove references to LIPIcs series (logo,
\nolinenumbers % for the final version

\def\RATIOP{302} % = 3*(3*2^5) + 12 + 2
\def\RATIOG{34}  % = 3*(11+\eps)+1 (without the \eps)
\def\RATIO{906}  % = 3*\RATIOP
\def\RATIOGO{35} % = \RATIOG+1

\usepackage{xspace}
\usepackage{tikz}

\usepackage{amsmath,amsthm}
\usepackage{graphicx,framed,calc,algorithm,algpseudocode}
\usepackage{mathrsfs}

\newtheorem{mythm}{Theorem}
\newtheorem{myconj}{Conjecture}
\newtheorem{myobs}{Observation}
\newtheorem{fact}{Fact}

\def\myparagraph#1{\medskip\noindent\textbf{#1}} % plus joli
\def\myparagraph#1{\subparagraph*{#1}} % pour LIPIcs

\newcommand{\set}[1]{\left\{{#1}\right\}}
\newcommand{\range}[2]{\set{{#1},\dots,{#2}}}

\newcommand{\ceil}[1]{\left\lceil{#1}\right\rceil}
\newcommand{\card}[1]{\left|{#1}\right|}

\newcommand{\LOCAL}{\textsf{LOCAL}\xspace}

\DeclareMathOperator{\MDS}{MDS}
\DeclareMathOperator{\asdim}{asdim}
\DeclareMathOperator{\poly}{poly}
\def\CC{\mathcal{C}}

% ...

\def\sfA{\mathsf{A}}
\def\sfB{\mathsf{B}}

\def\sC{\mathscr{C}}
\def\sD{\mathscr{D}}

\def\bbN{\mathbb{N}}
\def\minDS{\textsc{Minimum Dominating Set}\xspace}
\def\eps{\varepsilon}
\def\emo{{\eps}^{-1}}

\let\le\leqslant
\let\ge\geqslant
\let\leq\leqslant

\let\emptyset\varnothing

\title{Distributed Approximation Algorithms for Minimum Dominating Set in Locally Nice Graphs}

\titlerunning{}

\author{Marthe Bonamy}%
{LaBRI, University of Bordeaux, CNRS, France}%
{marthe.bonamy@u-bordeaux.fr}%
{https://orcid.org/0000-0001-7905-8018}{}

\author{Cyril Gavoille}%
{LaBRI, University of Bordeaux, CNRS, France}%
{gavoille@labri.fr}%
{https://orcid.org/0000-0003-3671-8607}{}

\author{Timoth{\'e} Picavet}%
{LaBRI, University of Bordeaux, CNRS, France}%
{timothe.picavet@u-bordeaux.fr}%
{https://orcid.org/0000-0002-7129-0127}{}

\author{Alexandra Wesolek}%
{Institut für Mathematik, Technische Universit{\"a}t Berlin, Germany}%
{wesolek@tu-berlin.de}%
{https://orcid.org/0000-0003-4841-5937}{}

\Copyright{Marthe Bonamy, Cyril Gavoille, Timoth{\'e} Picavet, and Alexandra Wesolek}

\authorrunning{M. Bonamy et al.}



\begin{document}
\maketitle

\begin{abstract}
 We give a new, short proof that graphs embeddable in a given Euler genus-$g$ surface admit a simple $f(g)$-round $\alpha$-approximation distributed algorithm for \minDS (MDS), where the approximation ratio $\alpha \le \RATIO$. Using tricks from Heydt et al. [\textit{European Journal of Combinatorics} (2025)], we in fact derive that $\alpha \le \RATIOG +\eps$, therefore improving upon the current state of the art of~$24g+O(1)$ due to Amiri et al. [\textit{ACM Transactions on Algorithms} (2019)]. It also improves the approximation ratio of~$91+\eps$ due to Czygrinow et al. [\textit{Theoretical Computer Science} (2019)] in the particular case of orientable surfaces.

% "... Smith et al. [Journal of Very Important Results, 1, 374 (2012)]"

All our distributed algorithms work in the deterministic \LOCAL model. They do not require any preliminary embedding of the graph and only rely on two things: a \LOCAL algorithm for MDS on planar graphs with ``uniform'' approximation guarantees and the knowledge that graphs embeddable in bounded Euler genus surfaces have asymptotic dimension~$2$. % (for any fixed genus).

More generally, our algorithms work in any graph class of bounded asymptotic dimension where ``most vertices'' are locally in a graph class that admits a \LOCAL algorithm for MDS with uniform approximation guarantees.
\end{abstract}

%\newpage

%%%%%%%%%%%%%%%%%%%%%%%%%%%%%%%%%%%%%%%%%%%%%%%%%%%%%%%%%%%%%%%%%%%%%%%%%

\section{Introduction}
\label{sec:intro}

% Cyril: 
% \mDS -> problem name
% minimum dominating set -> a special subset of vertices

\minDS (MDS) is a famous minimization problem on graphs, known to be NP-complete even in cubic planar graphs~\cite{GJ79,KYK80}. The goal is to find a smallest subset of vertices of the input graph that intersects all radius-1 balls of the graph.

In this paper, we focus our attention on distributed algorithms that can approximate MDS on the graph underlying the topology of the network. Due to the covering property of a dominating set, the problem and its variants (like connected dominating set~\cite{WAF02}) get increasingly more attention in Networking, in particular for mobile and ad-hoc networks. Not only are mobile networks important for point-to-point communications, but specialized ad-hoc networks, such as sensor networks, are important for environmental monitoring tasks. We refer to \cite{KW05,KWZ09} for extended discussions about the motivations of MDS for routing purposes of such networks.

We consider the standard \LOCAL model. In this model, popularized by Linial~\cite{Linial87}, nodes of the underlying network $G$ work in synchronous rounds of communications between direct neighbors of $G$ and are performed with messages of unlimited size. As the aim of this model is to study the ``local nature'' of a problem, the main complexity measure is the \emph{round complexity}, the maximum number of rounds of communication to perform a given task in the network.

In this model, approximating MDS up to a constant factor in general $n$-vertex graphs is known to require $\Omega(\sqrt{\log{n}/\log\log{n}}\,)$ rounds~\cite{KMW16,CL20}. % \cyril{We do not care about the bound expressed as the maximum degree, since we will not speak about this parameter in the paper. The goal is to say that this is not bounded.}
On the positive side, it is possible to $(1+\eps)$-approximate MDS for any graph in $\poly(\emo\log{n})$ rounds by combining the techniques of~\cite{GKM17} and of~\cite[Corollary~3.11]{RG20}.

For specific graphs, better round complexities can be achieved; cf.~\cite{LPW13,Suomela13} for a large collection of results (including unit-disk graphs, and planar graphs as a basis of the Gabriel graph model widely used in ad-hoc networks~\cite{WY07}). For instance, $O(\log^*{n})$ rounds suffice for planar graphs~\cite{CHW08a}, or more generally for $K_t$-minor-free graphs~\cite{CHW18} and for sub-logarithmic expansion graphs~\cite{ASS19}, and $o(\log^*{n})$ rounds are not sufficient to get a $(1+\eps)$-approximation of MDS on cycles~\cite{CHW08a} or unit-disk graphs~\cite{LW08}. So, achieving constant-round algorithms must be at the price of relaxing the $(1+\eps)$ approximation ratio.

For planar graphs, the quest for constant-approximation and constant-round algorithms seems to start with \cite{LOW08}, with a $126$-approximation. Since then, a long line of research has been aimed at improving this ratio. In short, the best approximation ratio to date is $11 + \eps$, due to~\cite{HKOSV25}. We refer to the nice survey of~\cite{HKOSV25} for planar graphs and sub-families, including lower bounds.

% There are $(1+\eps)$-approximation with poly-logarithmic number of rounds $O(\log^{28.7}{n}\cdot\log\log{n}\cdot\log^*{n})$ \cite{CHS06}.

In the meantime, the quest for constant-approximation ratio and constant-round distributed algorithms has been proposed for larger classes of graphs. However, very few examples are known of graph classes $\sC(p)$, depending of some fixed parameter $p$, where \minDS can be solved by a \LOCAL algorithm with round complexity $f(p)$ and truly-constant approximation ratio (independent of $p$). 

Larger graph classes that make good candidates for extending planar graphs are $H$-minor-free graphs, with some graph $H$ not limited to $K_5$ or $K_{3,3}$. For $K_p$-minor-free graphs, we only know of an exponential approximation ratio~\cite{KSV21,HKOSV25}, whereas a linear approximation ratio could be possible. For $K_{3,p}$-minor-free graphs~\cite{HKOSV25}, the dependency is much better (it is linear in~$p$), but still not truly constant. Very recently, \cite{BGPW25} show that, for $K_{2,p}$-minor-free graphs, the approximation ratio is~$50$ regardless of the value of $p$. %This has been extended to every any excluded minor of pathwidth at most two~\cite{}.

% graphs on surface

\def\rmM{\mathrm{M}}

% [CdV10]: ... In this paper, unless noted otherwise, S is a compact, connected surface without boundary; g denotes its Euler genus. Thus, if S is orientable, g ≥ 0 is even, and S is a sphere with g/2 handles attached; if S is non-orientable, S is a sphere with g ≥ 1 disks replaced by Möbius strips.

% [CPSTY12]: ... All graphs in this paper are finite, undirected and simple. By a surface we mean a compact, connected 2-dimensional manifold with empty boundary. The classification theorem of surfaces (see e.g. [16]) states that each surface is homeomorphic to either S_g, the surface obtained from the sphere by adding g handles, or N_k, the surface obtained from the sphere by adding k cross-caps. Thus S_0 = N_0 is the sphere, S_1 is the torus, N_1 is the projective plane and N_2 is the Klein bottle ... Heawood [11] proved that if Σ is not the sphere, then every graph in Σ is t-colorable as long as t ≥ H(Σ) := ⌊(7 + √24γ + 1)/2⌋, where γ is the Euler genus of Σ, defined as γ = 2g when Σ = S_g and γ = k when Σ = N_k. 

On the other end of the spectrum, the narrowest way to meaningfully extend planar graphs is through embeddable graphs\footnote{I.e., that can be drawn on the surface without edge crossing, like planar graphs for the sphere.} on surfaces\footnote{That are compact, connected $2$-dimensional manifold without boundary.} of Euler genus $g$. Such surfaces can be obtained from a sphere by adding $g/2 \ge 0$ handles\footnote{By adding a handle, we mean that two disjoint disks of the sphere are replaced by a cylinder.} (if they are orientable) or by adding $g\ge 1$ cross-caps\footnote{By adding a cross-cap, we mean that a disk of the sphere is replaced by a M{\"o}bius strip.} (if they are nonorientable). The \emph{Euler genus} of a graph $G$ is the minimum number $g$ such that $G$ embedds on a surface of Euler genus~$g$. In this context, \cite[Theorem~3.11]{ASS19} proposed an $(24g+O(1))$-approximation algorithm with constant-round complexity (actually linear in $g$). For graphs of orientable genus $g$, i.e., embeddable on an orientable surface of genus $g$, \cite{CHWW19} designed a constant-round algorithm that returns a dominating set of size at most $91\rmM + 76g - 66$, where $\rmM$ is the optimal size. By adding a simple extra brute-force step in their algorithm, we can convert it into a $(91+\eps)$-approximation\footnote{\label{foot:91+eps}Observe that the diameter of each connected component of the graph is at most $3\rmM$. So, each vertex check whether the diameter $D$ of its component is less than $B$, for $B = 3/\eps \cdot (76g-66)$. This can be done by collecting its radius-$B$ neighborhood in $B = O(g/\eps)$ rounds. If true, the vertex locally brute forces an optimal solution for its component that, which it already knows about. If false, the vertex applies the approximate algorithm. In that case, $B\le D$ and $D \le 3 \rmM$, which together imply $76g-66 \le \eps \rmM$. Thus, the returned set has size at most $91\rmM + 76g-66 \le (1+\eps)\rmM$.}, so with a ratio independent of $g$. However, the result of \cite{CHWW19} does not transfer to graphs of bounded Euler genus, as for any $g$, there are graphs embeddable on the projective plane\footnote{For instance, a cycle with $n = 2g+6$ vertices where opposite vertices are connected by an extra edge, cf. \cite{ABY63}.} that cannot be embedded on any orientable surface of genus $g$ (so of orientable genus at least $g+1$). So, for these graphs, the approximation ratio will be $91+\eps$ but the number of rounds possibly unbounded (possibly linear in the number of vertices).

Graphs of Euler genus-$g$ are included in $K_{3,2g+3}$-minor-free graphs, because the Euler genus of $K_{3,2g+3}$ is at least $g+1$ (cf. \cite[Theorem~4.4.7]{MT01}) and the class of Euler genus-$g$ graphs is closed under taking minors. We observe that the approximation ratio for $K_{3,p}$-minor-free graphs, for $p = 2g+3$, due to \cite[Theorem~2.2]{HKOSV25} provides an approximation ratio of $O(\sqrt{g})$. Indeed, the ratio of their algorithm is precisely $(2 + \eps)(2\nabla_1 + 1)$, where $\nabla_1$ is the maximal edge density of a depth-$1$ minor of any graph of the class. As we can check\footnote{Indeed, consider any $n$-vertex graph $G$ of Euler genus $g$, $\nabla_1(G)$ is no more than its edge density, and thus $\nabla_1(G) \le \card{E(G)}/n \le 3 + 3g/n$ from Euler's formula. As $\card{E(G)} \le \binom{n}{2}$, $\nabla_1(G) \le \min\set{(n+1)/2, 3g/n + 3}$ which is no more than $\sqrt{3g/2} + 3$.}
% Cyril: Let us check that M = min((n+1)/2,3+3g/n) <= 3+x*s where s=sqrt(g) and x=sqrt(3/2). Case 1: n <= 5 + 2*x*s. Then, M <= (n+1)/2 <= 3 + x*s, and we are done. Case 2: n >= 6 + 2*x*s. Then, M <= 3+3g/n <= 3 + 3g/(6+2*x*s). We need to check that 3 + 3g/(6+2*x*s) <= 3 + x*s, for x=sqrt(3/2). This is equivalent to check that 3g <= x*s*(6+2*x*s). We have x*s*(6+2*x*s) = 2*x^2*s^2 + 6*x*s = 2*(3/2)*g + 6*x*s = 3g + w, where w = 6*sqrt(3g/2) >= 0 for all g>=0. We are done in this case too.}
that $\nabla_1 \le \sqrt{3g/2} + 3$, the approximation ratio is at most $4\sqrt{3g/2} + 14 + \eps \approx 4.9 \sqrt{g} + O(1)$. As it may occur that $\nabla_1 \ge \sqrt{3g/2} - O(1)$ for some Euler genus-$g$ graph\footnote{Consider a clique $K_n$ for some $n = \sqrt{6g} + O(1)$, more precisely such that $\ceil{(n-3)(n-4)/6} = g$, where each edge is replaced by a path of length three. Such a graph has same Euler genus as $K_n$, which is~$g$ (cf.~\cite[Theorem~4.4.5]{MT01}), and has $K_n$ has depth-$1$ minor. Thus, $\nabla_1(G) \ge \card{E(K_n)}/n = (n+1)/2 \ge  \sqrt{3g/2} - O(1)$.}, the best known approximation ratio for MDS in Euler genus-$g$ graphs is $\Theta(\sqrt{g})$, and $91+\eps$ for orientable genus-$g$ graphs. 

\myparagraph{Our contributions.}

In the remaining of the paper, by ``\LOCAL algorithm'', we mean deterministic distributed algorithm in the \LOCAL model.

\begin{itemize}

\item On the conceptual side, we propose a generic algorithm~$\sfB$, that takes as input any ``uniform'' approximation \LOCAL algorithm $\sfA$ suitable for some graph class $\sC$, and that converts $\sfA$ into an approximation \LOCAL algorithm for a class of graphs that are almost ``locally'' in $\sC$ and of bounded asymptotic dimension. (See \cref{prop:local_approx_w_errors}.)

\item We give a simple proof and a simple algorithm~$\sfA$ for planar graphs that is a ``uniform'' $\RATIOP$-approximation in~$5$ rounds. (See \cref{prop:uniform_planar}.)

\item As a consequence and as a first technical contribution, we obtained, for Euler genus-$g$ graphs, a simple $\RATIO$-approximation algorithm with round complexity $C(g)$, for some function $C$. The algorithm is derived from an extension of an intriguing proposition due to \cite{BGPW25} (see also \cref{prop:local_approx_w_errors}), combined with the fact that such graphs have an asymptotic dimension at most two.

\item Using the fine analysis of \cite{HKOSV25} showing that their $(11+\eps)$-approximation algorithm is ``uniform'' (see \cref{obs:localplanar}), the ratio $\RATIO$ is decreased to $\RATIOG+\eps$, improving previous ratios that were $\Omega(\sqrt{g})$ at their best, and also the ratio $91+\eps$ for graphs of bounded orientable genus. (See \cref{th:bounded_genus}.)

\item More generally, our algorithms work in for ``locally nice'' graphs, that is any graph class of bounded asymptotic dimension where ``most vertices'' are locally in a graph class that admits a LOCAL algorithm for MDS with uniform approximation guarantees. In particular, for classes of asymptotic dimension $d$ that are locally of Euler genus $g$, we show an approximation factor of $\RATIOGO(d+1)$ with a round complexity of $C(g,d)$, for some function $C$. (See \cref{th:locally_bounded_genus}.)

\end{itemize}

Unlike similar notions, such as the classical network decomposition~\cite{AGLP89}, no preliminary construction is required for any of our algorithms. Obviously, they depend on the parameters of the graph classes on which they are applied, like the Euler genus $g$ and/or the asymptotic dimension of the class with its control function. A potential weakness of our algorithms is that it may use large messages, as is allowed in the \LOCAL model. However, we note that the extra step to convert the approximation algorithm of \cite{CHWW19} requires also large messages (see \cref{foot:91+eps}). We also use a brute-force algorithm for computing an optimal MDS on bounded radius part of the graph. Recall that the \LOCAL model focuses on the communication complexity of the algorithm, i.e., the number of rounds it takes to execute it in the worst case. In this model, each vertex has unbounded computational power. Thus, it allows us to brute-force optimal solutions on small-diameter graphs. However, we observe that the graphs we consider have bounded local treewidth -- that is, the treewidth of any radius-$r$ ball depends on $r$ and on $g$, the Euler genus of the ball. This structural property allows every vertex to calculate in polynomial time an optimal solution for MDS in its constant-radius neighborhood. More generally, this holds for any problem expressible in monadic second-order logic~\cite{Courcelle90}. Therefore, computational power does not seem to be a real limitation here.
%\tim{This is actually for MSO}

%\cyril{It is not clear that the 91 bound of \cite{CHWW19} works for "Euler genus". They say that "bounded genus $g$" graphs exclude $K_{3,4g+3}$, whereas \cite{BBEGx24} say explicitly (above Corollary~1.10) that Euler genus-$g$ graphs exclude $K_{3,2g+3}$. Generally speaking, the move from $g$ to $2g$ is due to orientable genus vs. Euler genus, the latter being more general.}

%\cyril{For the exact control function of Euler genus graphs, this is derived from the $2$-dimensional control function of $K_{3,p}$-minor free graphs in \cite[Theorem 1.9]{BBEGx24}. We can extract from Lemma 7.4 its control function, TODO. In fact, when we plug the $11+\eps$ results, we loose the exact round complexity, because they only give a round complexity of $C(\eps)$.}

%\myparagraph{State of the art.} ...

%Marthe : on copie-colle la table de PODC ?

\myparagraph{Organization  of  the  paper.} Notation and main notions are introduced in \cref{sec:prelim}: this notably includes uniform approximation and asymptotic dimension. \cref{sec:main} presents the main technical proposition for our algorithms, in particular the generic algorithm~$\sfB$ and the statements of the main theorems. \cref{sec:propwitherrors} and \cref{sec:localplanar} are dedicated to the proof of the key proposition using asymptotic dimension and to the fact that planar graphs admit a constant-uniform constant-approximation algorithm in constant rounds, respectively.

%%%%%%%%%%%%%%%%%%%%%%%%%%%%%%%%%%%%%%%%%%%%%%%%%%%%%%%%%%%%%%%%%%%%%%%%%

\section{Preliminaries}\label{sec:prelim}

\myparagraph{Graphs and domination.}

Consider a graph $G$, a vertex $v$ of $G$, $r \in\bbN$, and a subset $S\subseteq V(G)$. We denote by $N^r[v]$ the radius-$r$ ball centered at $v$ in $G$, that is, the set of all vertices at distance at most $r$ from $v$ in $G$. For short, we denote by $N[v]$ the set $N^1[v]$, and by $N(v)$ the set $N[v]\setminus\set{v}$ (the set of neighbors of $v$). By extension, $N^r[S] = \bigcup_{v \in S} N^r[v]$. Lastly, we denote by $G[S]$ the subgraph of $G$ induced by all the vertices of $S$. The \emph{weak diameter in $G$} of a subgraph $H$ of $G$ is the maximum distance in $G$ between any two vertices of $H$.

A subset $X\subseteq V(G)$ \emph{dominates} $S$ if $S\subseteq N[X]$. We denote by $\MDS(G,S)$ the minimum size of a subset of $V(G)$ dominating $S$. Note that $\MDS(G,V(G))$, denoted by $\MDS(G)$ for short, is the size of a minimum dominating set for $G$.

\myparagraph{Locally nice graphs.}

Consider two graph classes $\sC$ and $\sD$, and let $T \in\bbN$.

We say that $\sD$ is \emph{$T$-locally-$\sC$} if for every $G \in\sD$ and $u\in V(G)$, $G[N^T[u]]\in\sC$.

For a graph $G \in\sD$, the \emph{$T$-error set} of $G$ w.r.t. $\sC$ is the vertex set defined by $X = \set{u\in V(G) \mid G[N^T[u]] \notin \sC}$. In other words, this is the set of vertices that are witnesses for $G$ not being $T$-locally-$\sC$. Clearly, if $G$ has no $T$-errors, that is $X = \emptyset$, then $G$ is $T$-locally-$\sC$.

\myparagraph{Uniform approximation.}

For a \LOCAL algorithm $\sfA$ that returns a subset of vertices, we denote by $\sfA(G)$ the set returned by $\sfA$ when run on $G$. %The first important notion is that of \emph{uniformity} of a \LOCAL approximation algorithm for \minDS.

A \LOCAL algorithm $\sfA$ is a \emph{$k$-uniform} $\alpha$-approximation for \minDS in a hereditary\footnote{That is a class of graphs closed under vertex deletion.} class of graphs $\sC$, if for every $G\in\sC$ and $S\subseteq V(G)$, $\card{\sfA(G)\cap S} \le \alpha \cdot \MDS(G,N^k[S])$.

Unless explicitly said, we will assume for convenience that $k\ge 1$, as otherwise only very restricted graph class $\sC$ can support $0$-uniform constant-approximation algorithms for MDS. For example, consider a depth-2 tree $T$, with $\alpha + 1$ vertices at depth-1, each having $\alpha^2 + 3$ neighbors at depth-2. Clearly, $\MDS(T) = \alpha+1$. Moreover, any $\alpha$-approximation algorithm $\sfA$ in $T$ has to select all depth-1 vertices. If not, each non-selected vertex at depth-1 would force to select $\alpha^2+3$ extra vertices at depth-2, creating by this way a solution with at least $(\alpha-1) + (\alpha^2+3)$ vertices, which is strictly more than $\alpha \cdot \MDS(T)$. However, for a subset $S$ composed of all depth-1 vertices of $T$ (so with $\card{S} = \alpha+1$), we have on one side $\card{\sfA(G) \cap S} \ge \alpha+1$, whereas $\MDS(G,N^0[S]) = \MDS(G,S) = 1$ (considering the root), so $\card{\sfA(G) \cap S} \not\le \alpha \cdot \MDS(G,N^0[S])$. Therefore, $\sfA$ cannot be $0$-uniform, and more generally, every graph class $\sC$ including trees has no $0$-uniform $\alpha$-approximations for MDS, for every constant $\alpha$.

\myparagraph{Asymptotic dimension.}

A graph $G$ with a spanning subgraph $H$ is \emph{$k$-colorable $\delta$-bounded in $H$} if each vertex of $G$ can be assigned a color from $\range{0}{k-1}$ % Cyril: I changed to {0..k-1} (instead of {1..k}) as in subsequent proofs, we will use colors in {0..d}
so that the distance in $H$ between any two vertices taken in a monochromatic path of $G$ is at most $\delta$. In other words, all monochromatic connected components of $G$ must have a weak diameter in $H$ at most $\delta$. % Cyril: next is an example to stress that $k$-coloring here is not the usual proper $k$-coloring.
A proper $k$-coloring of $G$ is nothing else than a $k$-coloring $0$-bounded in $H$. If $H$ has diameter $D$, then it is $1$-colorable $D$-bounded in $H$.

The asymptotic dimension of a graph class $\sC$ is a non-negative integer denoted by $\asdim(\sC)$, and related to the $k$-coloring boundedness of the $r$-th power of $G$. It is a notion closely related to \emph{sparse partitions} and \emph{weak sparse covers} (see~\cite[§1.8]{BBEGx24} for precise connection between these notions). Recall that the $r$-th power of $G$, denoted by $G^r$, is the graph obtained from $V(G)$ by adding edges between pair of vertices at distance at most $r$ in $G$. 

Rather than giving its standard definition, we prefer the following equivalent one\footnote{In the original Proposition~1.17 of \cite{BBEGx24}, $G^r$ is $(d+1)$-colorable $f(r)$-bounded in $G^r$, instead of in $G$. By following this original definition, the function $f$ is not a control function, whereas in \cref{prop:asdim_def} it is.}:

\begin{proposition}[{\cite[Proposition 1.17]{BBEGx24}}]\label{prop:asdim_def}
For every graph class $\sC$, $\asdim(\sC) \le d$ if and only if there exists a function $f: \bbN \to \bbN$, called \emph{control function}, such that for every $G \in \sC$ and integer $r\ge 1$, $G^r$ is $(d+1)$-colorable $f(r)$-bounded in $G$. %\cyril{Does not work with $r=0$.} \cyril{Je me suis gouré: d'après la preuve de Porp. 1.17, la fonction de controle pour $\sC$, c'est $r\cdot f(r)$ et pas $f(r)$.}
\end{proposition}

% Observe that, by definition, monochromatic connected components with the same color in $G^r$ are pairwise at distance $>r$ in $G$. This is a key feature of graphs of asymptotic dimension~$d$, as it implies that any LOCAL algorithm that runs in~$r$ can be in conflict only with at most~$d$ ...

Any associated control function $f$ that makes $\asdim(\sC) = d$ is also called a $d$-dimensional control function of $\sC$. The Assouad-Nagata dimension of $\sC$ has a similar definition, except that the control function must be linear. A crude example is that paths have Assouad-Nagata dimension~$1$, because, for each $r\ge 1$, it suffices to assign to each path $P$ a color with each segment of $r$ consecutive vertices, and continuing by alternating with two colors. As two segments of the same color are at distance $>r$ in $P$, monochromatic connected components of $P^r$ have weak diameter in $P$ at most $r-1$ (corresponding to the diameter in $P$ of each segment). In other words, $P^r$ is $2$-colorable $(r-1)$-bounded in $P$, and so $f(r) = r-1$ is a $1$-dimensional control function of the class of paths.

%%%%%%%%%%%%%%%%%%%%%%%%%%%%%%%%%%%%%%%%%%%%%%%%%%%%%%%%%%%%%%%%%%%%%%%%%

\section{Tools and main results}\label{sec:main}

Our algorithm was inspired by the following intriguing property~\cite{BGPW25}. Roughly speaking, any approximate algorithm $\sfA$ is going to perform well in a target graph class $\sD$ that is locally-$\sC$, as long as $\sfA$ is $k$-uniform on $\sC$ and $\sD$ has small asymptotic dimension. 

\begin{restatable}[{\cite[Proposition~3.1]{BGPW25}}]{proposition}{PropI}\label{prop:local_approx}
    Let $\sfA$ be a $k$-uniform $\alpha$-approximation \LOCAL algorithm for \minDS in a hereditary class of graphs $\sC$ with round complexity $r \ge 1$. Let $\sD$ be a graph class with $d$-dimensional control function $f$, and that is\footnote{The original statement claimed erroneously that $\sD$ being $t$-locally-$\sC$ for $t = f(2k+1)$ is enough, instead of the correct $t = f(2k+2) + r$. This has no impact, as in practice we use \cref{prop:local_approx} in \cref{th:locally_bounded_genus} only for $k=7$.} $(f(2k+2)+r)$-locally-$\sC$. Then $\sfA$ is also an $\alpha (d+1)$-approximation algorithm on $\sD$. 
\end{restatable}

We extend \cref{prop:local_approx} to allow for locally constrained errors.

%\tim{I think we also should say $k\geq 0$. We need it for the proof, otherwise $\sfB$ is $1$-uniform and not $0$-uniform. Or in generality, we can say that $\sfB$ is $\max\set{1,k}$-uniform, but that's ugly.} \cyril{I make a comment in its definition.}
\begin{restatable}{proposition}{LocalToGlobalWithErrors}\label{prop:local_approx_w_errors}
    Let $\sfA$ be a $k$-uniform $\alpha$-approximation \LOCAL algorithm for \minDS in a hereditary class of graphs $\sC$ with round complexity $r$. %\cyril{I would write "... with round complexity $r$." So, do we need in the proof that $r\neq 0$? The new proof does not requires $r\ge 1$ anymore.}
    Let $\sD$ be a graph class with a $d$-dimensional control function $f$. Then, there exists a $k$-uniform $(\alpha (d+1)+1)$-approximation algorithm $\sfB$ on $\sD$ with round complexity $T+\delta+2$, where $T = f(2k+2) + \max\set{k+1,r}$ and where $\delta$ is the maximum weak diameter in any $G\in\sD$ of any connected component of $G[N^2[X]]$, $X$ being the $T$-errors of $G$ w.r.t. $\sC$. %\cyril{It would be nice to say that $\sfB$ is also $k'$-uniform for some $k'$, since we could re-apply it (as long as it is hereditary and still of small asdim.}
    Furthermore, if $\sD$ is $T$-locally-$\sC$ (i.e., there is no $T$-errors), then the approximation ratio of $\sfB$ is $\alpha (d+1)$.
    %\cyril{I suggest to add: ``Furthermore, if $\sD$ is $T$-locally-$\sC$ (i.e., $X = \emptyset$), then the approximation ratio of $\sfB$ is $\alpha (d+1)$.'' This allows us to catch for free the uniform property of $\sfB$ if $\sD$ is $T$-locally-$\sC$.}
\end{restatable}

The proof is an adaptation of that of Proposition~\ref{prop:local_approx}, and we include it in~\cref{sec:propwitherrors}. However, we already describe the corresponding algorithm $\sfB$.

\begin{algorithm}
\caption{Generic algorithm $\sfB$ from \cref{prop:local_approx_w_errors}.}\label{algo:MDS}
\begin{algorithmic}[1]

\Require A $k$-uniform $\alpha$-approximate \LOCAL algorithm $\sfA$ for \minDS in a hereditary class of graphs $\sC$ with round complexity $r$, a graph class $\sD$ with a $d$-dimensional control function $f$, $G\in\sD$ with $T$-errors $X$ where $T = f(2k+2) + \max\set{k+1,r}$. %\cyril{I think it is more clear that we speak in term of "$d$-dimensional control function $f$" when we care about about $f$, and of "asymptotic dimension $d$" when we do not care about $f$. In fact here (and in the above proposition) we really don't care about $d$ for the input. For the input we care about the control function. Whereas we care of $d$ in the output (approx ratio).}

\Ensure A dominating set $S$ of $G$ such that $\card{S} \le (\alpha (d+1)+1) \cdot \MDS(G)$.

\State $S \gets \emptyset$

\State Each vertex $u$ computes $G[N^T[u]]$ and checks whether it belongs to $\sC$; as a consequence, it decides whether it belongs to $X$. 

%\cyril{I've removed "the next" in "for the next $r$ rounds". this is because to can run in parallel Step~2 and Step~3 to get a slighlty better running time.}
\State Each vertex $u$ runs $\sfA$ for $r$ rounds, and gets added to $S$ only if $u\in \sfA(G)\setminus X$.

\State $S \gets S \cup S'$, where $S'$ is a brute-forced minimum set of $G$ that dominates $V(G)\setminus N[S]$.

\end{algorithmic}
\end{algorithm}

To illustrate \cref{prop:local_approx_w_errors}, assume that $\sC$ and $\sD$ are respectively the class of planar and Euler genus-$g$ graphs, for a given $g\ge 1$. Moreover, assume that we are given a \LOCAL algorithm $\sfA$ with round complexity $r$ for \minDS that is a $k$-uniform $\alpha$-approximation on $\sC$, where $k,\alpha,r$ are small constants.

Euler genus-$g$ graphs exclude $K_{3,2g+3}$, and thus have asymptotic dimension at most two~\cite[Theorem~1.9]{BBEGx24}, and even Assouad-Nagata dimension at most two. So, there is a $2$-dimensional control function $f(x) \le c(g) \cdot x$, for some function $c$.

Consider a graph $G \in\sD$ with $T$-error set $X$, where $T = f(2k+2) + \max\set{k+1,r} = O_g(1)$. 

According to \cref{prop:local_approx_w_errors} applied on $\sC$, $\sD$ and $\sfA$, Algorithm $\sfB$ has approximate ratio $2(\alpha+1)+1 = O(1)$ for $G$ and has round complexity $T + \delta + 2 = O_g(1) + \delta$. To conclude that the round complexity of $\sfB$ is in $O_g(1)$, we show that:

\begin{claim}\label{claim:genus_additivity}
     $\delta < g\cdot(2T+5)$, for any $g\ge 1$.
\end{claim}

\begin{proof}
    To derive a contradiction, consider a connected component $H$ of $G[N^2[X]]$, and assume that $H$ has weak diameter $\delta \ge g \cdot (2T+5)$ in $G$. As $g\ge 1$, $X$ and $H$ are not empty. So, $H$ must contain some path $P$ such that the distance in $G$ between its endpoints, say $s$ to $t$, is $\delta$. In particular, for each $d \in\set{0,1,\dots,\delta}$, there must exist a vertex of $P$ that is at distance exactly $d$ in $G$ from~$s$. This is because the distance in $G$ from $s$, when moving along an edge of~$P$, is a function that can vary by at most one unit, and this distance goes from $0$ (at $s$) to $\delta$ (at $t$). For every $i \in\set{0,1,\dots,g}$, let $u_i$ be any vertex of $P$ at distance exactly $d_i =  i\cdot (2T+5)$ in $G$ from $s$. So $u_0 = s$, $u_g = t$, and all the $u_i$'s exist in $P$ (so in $H$) since $d_i \in\set{0,1,\dots,\delta}$. By definition of $H$, each $u_i$ intersects $N^2[X]$. So, for each $u_i$ one can select a vertex $x_i\in X$ at distance in $G$ at most~2 from $u_i$. By the triangle inequality, for all $0\le i<j\le g$, $x_i$ and $x_j$ are at distance in $G$ at least $(d_j-2) - (d_i+2) = (j-i) \cdot(2T+5) - 4 \ge 2T+1$. It follows that $N^T[x_i]$ and $N^T[x_j]$ are disjoint. By definition of $X$, the subgraphs $G[N^T[x_i]]$'s are not planar, thus of Euler genus $\ge 1$. By additivity of the Euler genera
    % Cyril{genera = plurial of genus.}
    of these $g+1$ pairwise disjoint subgraphs of $G$ (cf.~\cite[Theorem~4.4.3]{MT01}),
    % Cyril{In fact this theorem says that the Euler genus of a graph is the sum of its Euler genera of its blocks. So deleting edges in $G$, we can obtain one block of each subgraph.}
    we get a contradiction that $G$ has Euler genus $g$. Thus $\delta < g\cdot (2T+5)$ as claimed.
\end{proof}

We remark that \cref{claim:genus_additivity} is the only place where Euler genus is used. From the above discussion, to prove our approximation result on Euler genus-$g$ graphs (\cref{th:bounded_genus}), we simply need to design a $O(1)$-uniform $O(1)$-approximation algorithm $\sfA$ for planar graphs. In other words, we have made a reduction of $3\alpha$-approximation for Euler genus-$g$ graphs from $O(1)$-uniform $\alpha$-approximation algorithm for planar graphs.

\begin{restatable}{proposition}{PropUniformPlanar}\label{prop:uniform_planar}
    There is a $4$-uniform $\RATIOP$-approximation \LOCAL algorithm $\sfA$ for \minDS in planar graphs with round complexity~$5$.
\end{restatable}

We present a very short proof of~\cref{prop:uniform_planar} in~\cref{sec:localplanar}. A direct application of \cref{prop:local_approx_w_errors} and of the above discussion, is that Algorithm~$\sfB$ is a $\RATIO$-approximation for Euler genus-$g$ graphs. However, it is possible to obtain a much better ratio for planar graphs, with significantly more work. Although it is not explicitly stated in their paper, we can observe the following.

\begin{observation}\label{obs:localplanar}
 The algorithm described in~\cite[Theorem~2.3]{HKOSV25} is a $7$-uniform $(11+\eps)$-approximation \LOCAL algorithm for \minDS in planar graphs with round complexity $C(\eps)$, for some function $C$.
\end{observation}%TODO : vérifier 7 + regarder où y'a le epsilon

From~\cref{obs:localplanar} and~\cref{prop:local_approx_w_errors}, and from the above discussion, we immediately obtain our first main result:

\begin{theorem}\label{th:bounded_genus}
    For any $\eps>0$, for any $g$, there is a $7$-uniform $(\RATIOG +\eps)$-approximation \LOCAL algorithm for \minDS in Euler genus-$g$ graphs with round complexity $C(\eps,g)$ for some function $C$.
\end{theorem}%TODO : vérifier fonction de contrôle, est-ce que c'est bien linéaire ?
%TODO: est-ce bien $O(\eps^{-1})$-uniform?

We can even go one step further by applying \cref{prop:local_approx} to the uniform approximation \LOCAL algorithm derived from~\cref{th:bounded_genus}, and the class $\sC$ of Euler genus-$g$ graphs. We immediately obtain our second main result:

\begin{theorem}\label{th:locally_bounded_genus}
    Let $\sC$ be the class of Euler genus-$g$ graphs, and $\sD$ be a graph class of asymptotic dimension $d$ that is $r$-locally-$\sC$ for some $r$ large enough. Then, there is a $\RATIOGO (d+1)$-approximation \LOCAL algorithm for \minDS in $\sD$ with round complexity $C(d,g,r)$ for some function $C$.
\end{theorem}

%\cyril{I prefer to round to $35(d+1)$ instead of putting $O(d)$ or $(34+\eps) (d+1) = (34+\eps)d + 34 + \eps$ and add "for every $\eps$ ..."}\marthe{Yeah agree}

% Cyril: we have $d\ge 2$. Assume $\eps=0.1$, then the approx ratio is $34.1 d + 34.1$. $35d = (34+\eps)d + d-\eps = 34d + d-0.1$.

%%%%%%%%%%%%%%%%%%%%%%%%%%%%%%%%%%%%%%%%%%%%%%%%%%%%%%%%%%%%%%%%%%%%%%%%%

\section{Self-contained proof of Proposition~\ref{prop:local_approx_w_errors}}\label{sec:propwitherrors}

\LocalToGlobalWithErrors*

\begin{proof}
    %\cyril{For setminus operator, please choose between $X\setminus Y$ (that I prefer) and $X-Y$. Please check and uniformize in all the paper. }\tim{$\setminus$ is for sets, $-$ is for graphs.}
    Let $G\in \sD$. As $\sD$ has $d$-dimensional control function $f$, the graph $G^{2k+2}$ admits a $(d+1)$-coloring $f(2k+2)$-bounded in $G$. Fix such a coloring, and, for each $i\in\set{0,1,\dots,d}$, let $C_i$ be the set of color-$i$ vertices in $G^{2k+2}$ (and also in $G$). We denote by $\CC_i$ the set of connected components of $G^{2k+2}[C_i]$. 
    So, by definition, all components of $\CC_i$ have weak diameter in $G$ at most $f(2k+2)$, and are pairwise at distance at least $2k+3$ in $G$ (because distinct components of $\CC_i$ cannot be adjacent in $G^{2k+2}$).
    
    We run Algorithm~$\sfB$ on $G$. It is easy to check that $\sfB(G)$ is a valid dominating set for $G$. By definition, $\card{\sfB(G)} \le |S| + |S'|$, where $S = \sfA(G)\setminus X$ (Step~3) and $S'$ is a brute-forced minimum set dominating $V(G)\setminus N[S]$ (Step~4). Clearly, $\card{S'} \le \MDS(G)$, and $\card{S'} = 0$ if there is no $T$-errors ($X=\emptyset$). Moreover, by definition of the coloring, $\card{\sfA(G)\setminus X} = \sum_{i=0}^d \card{(\sfA(G)\cap C_i)\setminus X}$. In other words,
    \begin{equation}\label{eq:B(G)}
    \card{\sfB(G)} ~\le~ \sum_{i=0}^d \card{(\sfA(G)\cap C_i)\setminus X} + \left\{%
        \begin{array}{ll}
        0 &\mbox{if $X=\emptyset$}\\
        \MDS(G) &\mbox{otherwise}
        \end{array}
        \right.
    \end{equation}
    
    To upper bound the term $\card{(\sfA(G)\cap C_i)\setminus X}$, consider some color $i\in\set{0,1,\dots,d}$ and some component $C\in \CC_i$ such that $C\not\subseteq X$. Consider some $v\in C\setminus X$, and denote $G_v = G[N^T[v]]$. Note that, by definition of $X$ and as $v\notin X$, $G_v \in\sC$.

    Along the proof, we will use several times the following fact:

    \begin{fact}\label{fact:dom}
        For any graph $H$ and $W \subseteq V(H)$, to dominate $W$ in $H$, it is enough to select vertices taken from $N[W]$.
    \end{fact}
    
    In the following, to differentiate the neighborhoods in $G$ and in $G_v$, we write them using either $N_G$ or $N_{G_v}$ notations.
    Notice that, for each non-negative integer $t\le \max\set{k+1,r}$, $G[N_{G}^t[C]]\subseteq G_v$. This is because $C$ has weak diameter in $G$ at most $f(2k+2)$ and thus $G[N_{G}^t[C]]$ has weak diameter in $G$ at most $f(2k+2) + t \le T$, by the choice of $T$. Moreover, as $\sC$ is hereditary, we have that $G[N_{G}^t[C]]\in \sC$.

    As $G[N^t_G[C]]\subseteq G_v$, vertices in $C$ have the same distance $t$-neighborhood in $G$ and $G_v$. In particular, $G[N^t_G[C]] = G_v[N^t_{G_v}[C]]$. From~\cref{fact:dom}, any minimum set of vertices of $G$ that dominates $N^k_G[C]$ is contained in $G[N^{k+1}_G[C]]$, and similarly if we replace $G$ by $G_v$. It follows that $\MDS(G,N^k_{G}[C]) = \MDS(G_v,N^k_{G_v}[C])$. Because Algorithm~$\sfA$ is a $k$-uniform $\alpha$-approximation on $\sC$, and $G_v\in \sC$, we get
    \[
        \card{\sfA(G_v)\cap C} ~\le~ \alpha \cdot \MDS(G_v,N_{G_v}^k[C]) = \alpha \cdot \MDS(G, N^k_G[C]) ~.
    \]
    Since $G[N^r_G[C]] = G_v[N^r_{G_v}[C]]$, Algorithm~$\sfA$, which runs in $r$ rounds, returns the same vertex set in $G$ and $G_v$ for vertices in $C$. Therefore, $\card{\sfA(G)\cap C} = \card{\sfA(G_v)\cap C}$. Combining with the previous inequality, we get
    \[
        \card{\sfA(G)\cap C} ~\le~ \alpha \cdot \MDS(G, N_G^k[C]),\quad\forall C\in\CC_i, C\not\subseteq X.
    \]
    Recall that any two connected components $C, C' \in \CC_i$ are at distance in $G$ at least $2k+3$. So $G[N^{k+1}_G[C]]$ and $G[N^{k+1}_G[C']]$ are disjoint subgraphs, and, as to dominate $V(G)$ we need to dominate $N^{k}_G[C]$ and $N^{k}_G[C']$, by disjunction of theses sets and by~\cref{fact:dom}, we have
    \def\ccix{\substack{C\in\CC_i\\ C\not\subseteq X}}
    \[
        \sum_{\ccix} \MDS(G, N_G^k[C]) ~\le~ \sum_{C\in \CC_i} \MDS(G, N_G^k[C]) ~\le~ \MDS(G,V(G)) ~=~ \MDS(G) ~.
    \]
    We remark that $\card{(\sfA(G)\cap C)\setminus X} = 0$ if $C \subseteq X$. So, by combining previous inequalities, we obtain the following. %\cyril{I'm wondering in the 2nd line below (when removing the "$\setminus X$"), whether we cannot remove vertices of $S'$ that must be in all $X$, isn't? so winning le +1 in the approx.}
    
    \begin{eqnarray*}
        \card{(\sfA(G)\cap C_i)\setminus X} &=& \sum_{C\in\CC_i} \card{(\sfA(G)\cap C)\setminus X} ~=~ \sum_{\ccix} \card{(\sfA(G)\cap C)\setminus X}\\
        &\le& \sum_{\ccix} \card{\sfA(G)\cap C} ~\le~ \alpha \sum_{\ccix} \MDS(G,N^k_G[C]) \\
        &\le& \alpha\cdot \MDS(G) ~.
    \end{eqnarray*}

    Putting in \cref{eq:B(G)}, we get that Algorithm~$\sfB$ is an $(\alpha(d+1)+1)$-approximation, and even an $\alpha(d+1)$-approximation if $X=\emptyset$.

    To see that $\sfB$ is also a $k$-uniform approximation it is sufficient to see that, for any $W \subseteq V(G)$, the bound obtained is $\card{\sfB(G)\cap W} \le \alpha (d+1)\cdot \MDS(G, N_G^k[W]) + \MDS(G, N_G[W])$ where the first part is due to running $\sfA$ and the second part is due to running the brute-force.
    Indeed, if the brute-force computed a set whose intersection with $W$ was smaller than $\MDS(G, N_G[W])$, we could replace it by the minimum dominating set of $N_G[W]$ in $G$ and obtain a smaller dominating set, a contradiction. Note that $\MDS(G, N_G[W]) \le \MDS(G, N^k_G[W])$ for $k\ge 1$, proving that $\sfB$ is a $k$-uniform with an approximate ratio $\alpha(d+1)+1$ (or $\alpha(d+1)$ if $X=\emptyset$, as the brute-force is not required in that case).
    
    Now, let us prove that Algorithm~$\sfB$ has the desired round complexity.
    \begin{itemize}
    
        \item Computing $X$ in Step~2 takes $T+1$ rounds.
    
        \item Running $\sfA$ in Step~3 takes~$r$ rounds.
    
        \item For Step~4, consider the set $S$ computed by $\sfB$ at Step~3, before the brute-force.
        Observe that the set $W = V(G)\setminus N_G[X]$ is dominated by $S$. This is because vertices of $W$ are at distance at least two from $X$ and thus vertices of $N_G[W]$ cannot be in $X$. Thus, $\sfA$ applies to all vertices of $N_G[W]$, and so $W$ is indeed dominated by $S$ (\cref{fact:dom}).
        Therefore, to dominate $N_G[X]$ it is enough to select a set $S'$ from $N^2_G[X]$ (\cref{fact:dom}) as done in Step~4.
        By assumption, the connected components of $N_G^2[X]$ have weak diameter in $G$ at most $\delta$ (obviously, if $X = \emptyset$, $\delta = 0$). Therefore, the brute-force (Step~4) will take at most $\delta+1$ rounds. 
        
    \end{itemize}
    Naively, Steps~2 and~3 together take $(T+1) + r$ rounds. As Algorithm~$\sfA$ is not guaranteed to work when executed on vertices of $X$, Step~3 must be run only after Step~2. However, both steps can be run in parallel as follows. We run $\sfA$ for $r$ rounds exactly and stop it just after (since its running time on a vertex of $X$ could result into more than $r$ rounds). Then, the decision to add $u$ in $S$ (if selected by $\sfA$) is delayed for the next $T+1 - r$ rounds (this is non-negative from the choice of $T$). In parallel, $G[N^T[u]]$ is computed and is checked to be in $\sC$ or not after $T+1$ rounds. So, after $\max\set{r,T+1} = T+1$ rounds, set $S$ in Step~3 has been completed.
    
    Step~4 takes $\delta+1$ steps, so that the total round complexity of $\sfB$ is $T+\delta+2$, completing the proof of \cref{prop:local_approx_w_errors}.
\end{proof}

%%%%%%%%%%%%%%%%%%%%%%%%%%%%%%%%%%%%%%%%%%%%%%%%%%%%%%%%%%%%%%%%%%%%%%%%%

\section{Self-contained proof of Proposition~\ref{prop:uniform_planar}}\label{sec:localplanar}

\PropUniformPlanar*

\begin{proof}
To present the algorithm (cf. \cref{algo:planar}) we first need to define the ``best'' among all minimum dominating sets of a set $X$ in a labelled graph $H$. Among all of those sets, we first discard those which contain some $v$ for which there exists $w$ with $N[v]\subsetneq N[w]$, and among the remaining ones, we declare ``best'' the one lexicographically smallest considering the labels of the vertices. Note that such a set always exists.

\begin{algorithm}
\caption{The algorithm $\sfA$ for dominating set on planar graphs.}\label{algo:MDSplanar}
\begin{algorithmic}[1]

\Require A planar graph $G$.

\Ensure A dominating set $D$ of $G$ such that $\card{D\cap S}\le \RATIOP \cdot \MDS(G, N^4[S])$ for any subset $S \subseteq V(G)$.

\State $D \gets \emptyset$

\State Each vertex $u$ computes $G[N^{4}[u]]$ and all minimum dominating sets of $N^3[S]$ in $G$. Among those sets, $u$ selects the ``best'' one as $D_u$.
%selects those which do not contain any $v$ for which there exists $w$ with $N[v]\subsetneq N[w]$, and among those names $D_u$ the one lexicographically smallest considering the labels of the vertices. 

\State Each vertex $u$ picks the vertex $v_u$ of smallest label in $D_u\cap N[u]$, and adds it to $D$. %Let $S$ be set of vertices taken to be in the solution at the end of this step.

\end{algorithmic}
\label{algo:planar}
\end{algorithm}

By construction, Algorithm $\sfA$ outputs a dominating set $D$ within~$5$ rounds. We need to argue a little more for the uniform approximation.

\begin{claim}
    Algorithm $\sfA$ is a $4$-uniform $\RATIOP$-approximation.
\end{claim}

\begin{proof}
Let $G$ be a planar graph, and let $D = \sfA(G)$ be the dominating set of $G$ returned by $\sfA$. Let $S \subseteq V(G)$, and let $X$ be the ``best'' minimum dominating set of $N^4[S]$ in $G$. Our goal is to obtain that $\card{D \cap S}\leq \RATIOP \card{X}$. Let the $D_u$'s and $v_u$'s be as defined in $\sfA$.

The first key observation is that given disjoint subsets  $A,B,C$ of $V(G)$ such that  $G = G[A \cup C]\cup G[B\cup C]$, that is, $C$ separates $A$ from $B$, if $G[B\cup C]$ has radius at most $3$, then for any vertex $u \in B$, the value of $D_u \cap B$ (hence of $v_u$) is completely determined by $D_u \cap C$. In other words, there is a unique best way to extend $D_u \cap C$ to $B$. Therefore, the set $B$ can only yield at most $2^{\card{C}}\cdot \MDS(G,B\cup C)$ different candidates for $v_u$ with $u\in B$.

% \cyril{multigraph and multiple edges (multiedge does not exist)}
We now discuss how to use that observation to reach the desired conclusion. Let $Y = N^4[S]\setminus X$. Since $X$ is a dominating set of $Y$, every vertex in $Y$ has a neighbor in $X$. We associate to each vertex $u \in Y$ an arbitrary vertex $w_u \in X \cap N(u)$. We consider the plane multigraph $H$ obtained from an arbitrary planar embedding of $G[X \cup Y]$ by contracting every edge $uw_u$ for $u \in Y$, but keeping any resulting loop or multiple edge. Note that $V(H) = X$ by construction. Let $H'$ be the plane multigraph obtained from $H$ by iteratively deleting any edge incident to a face of size $1$ or to two faces of size $2$. In particular, $H'$ may contain loops and multiple edges, but in the embedding they separate different parts of $V(H')$ or, in the case of multiple edges, are incident on at least one side to a face of size at least $3$. We observe that by Euler's formula, there are fewer than $2\cdot 3\card{V(H')} = 6\card{H} = 6\card{X}$ edges in $H'$. We let $Z$ be the subset of vertices of $Y$ which are involved in an edge of $E(H')$. Here, involved means that the edge corresponds in $G$ to an edge incident to them. Since every edge corresponds to at most $2$ vertices of $Y$, we have $\card{Z}\le 12 \card{X}$ hence $\card{\set{v_u \mid u \in Z}} \leq 12 \card{X}$.

It remains to bound $\card{\set{v_u \mid u \in Y \setminus Z}}$. Note that there are two types of vertices in $Y \setminus Z$. Some are vertices only involved in loops, and some are vertices only involved in loops and (at least one) edge incident to two faces of size $2$. In the first case, we can observe that such a vertex $u$ has no neighbor not in the closed neighbourhood of $w_u$, which must by the definition of ``best'' be the same as $v_u$. The second case is slightly more interesting. We consider a face $f$ in $H'$ which contained in $H$ some deleted edge where such a vertex $u$ was involved (in fact, such a face is unique, but that is irrelevant here). We observe that $f$ is a face of size $2$, consisting of two edges $e$ and $e'$ and two vertices $x_1$ and $x_2$. In $G$, the edge $e$ is incident to $s_1,s_2$ and $e'$ to $t_1,t_2$, where for each $i$, either $v_{s_i}=x_i$ or $s_i=x_i$ and similarly for $t_i$. Note that $\{x_1,x_2,s_1,s_2,t_1,t_2\}$ separates $u$ from the rest of $V(H')$, and that vertices involved in some edge inside $f$ in $H$ are all adjacent to $x_1$ or to $x_2$ hence a radius of at most $3$. Therefore, there are at most $2^{6} \cdot 2$ choices for $v_u$ if $u$ is involved in some deleted edge inside $f$. We can note that in the case where both $x_1$ and $x_2$ get selected in $D_u$ there is no extra vertex to add as all the inside is dominated. Therefore, we reduce the number of choices to $3\cdot 2^5 = 96$. Since there are at most $3\card{X}$ options for $f$, this adds up to $288\card{X}$.

All together, we get $2\card{X}$ for the set $X$ itself being possibly chosen (for example by the vertices only involved in loops), and its own choices, then $12\card{X}$ for the choices of $Z$, and finally $288\card{X}$ for the choices of the second case of vertices not in $Z$. This sums up to $\RATIOP\card{X}$, as claimed.
\end{proof}

This completes the proof of \cref{prop:uniform_planar}.
\end{proof}

\section{Conclusion}
\label{sec:conclusion}

We have proposed rather simple constant-round approximation \LOCAL algorithms for \minDS achieving, in their best, an approximation ratio of $\RATIOG + \eps$ for \minDS in Euler genus-$g$ graphs, improving the previous ratio of $\Omega(\sqrt{g})$, and also the bound of~$91+\eps$ for orientable surfaces. We believe that an even better ratio can be obtained by a finer analysis of our algorithms. We even believe that the best approximation ratio should be close to the one for planar graphs, and we propose:

\begin{conjecture}\label{conj:planar}
    For every $g \in\bbN$, there is a $7$-approximation \LOCAL algorithm for \minDS in Euler genus-$g$ graphs with round complexity $f(g)$, for some function $f$.
\end{conjecture}

Actually, our algorithms also apply to graph classes that are locally of bounded genus, the approximation ratio depending only linearly on the asymptotic dimension of the class. We left open the question of extending the techniques and algorithms developed in this paper to approximate other optimization problems in the \LOCAL model.

%\myparagraph{Acknowledgments:} ...

\newpage

\bibliographystyle{my_alpha_doi}
\input{my_patch_lipics.tex}

\bibliography{bibliography,biblio_cyril}

\end{document}

%% file: my_patch_lipics.tex
% Permet d'avoir des un bibliographystyle{my_alpha_doi}
% avec le style LIPIcs.
%
% C'est la copie de la redéfinition de la commande
% \thebibliography par lipics-v2021.cls où j'ai remplacé
% un \leftmargin8.5mm par \leftmargin\MYMARGIN que l'on
% peut fixer à 18.5mm ou moins.
%
% Cyril - 27/05/2025

\makeatletter
\ifdefined\MYMARGIN\else\def\MYMARGIN{18.5mm}\fi
\renewenvironment{thebibliography}[1]
  {\if@noskipsec \leavevmode \fi
   \par
   \@tempskipa-3.5ex \@plus -1ex \@minus -.2ex\relax
   \@afterindenttrue
   \@tempskipa -\@tempskipa \@afterindentfalse
   \if@nobreak
     \everypar{}%
   \else
     \addpenalty\@secpenalty\addvspace\@tempskipa
   \fi
   \noindent
   \rlap{\color{lipicsLineGray}\vrule\@width\textwidth\@height1\p@}%
   \hspace*{7mm}\fboxsep1.5mm\colorbox[rgb]{1,1,1}{\raisebox{-0.4ex}{%
     \normalsize\sffamily\bfseries\refname}}%
   \@xsect{1ex \@plus.2ex}%
   \list{\@biblabel{\@arabic\c@enumiv}}%
        {\leftmargin\MYMARGIN
         \labelsep\leftmargin
         \settowidth\labelwidth{\@biblabel{#1}}%
         \advance\labelsep-\labelwidth
         \usecounter{enumiv}%
         \let\p@enumiv\@empty
         \renewcommand\theenumiv{\@arabic\c@enumiv}}%
   \fontsize{9}{12}\selectfont
   \sloppy
   \clubpenalty4000
   \@clubpenalty \clubpenalty
   \widowpenalty4000%
   \sfcode`\.\@m\protected@write\@auxout{}{\string\gdef\string\@pageNumberStartBibliography{\thepage}}}
\makeatother